\documentclass[reqno, centertags, 12pt, draft]{article}

\usepackage{amscd}
\usepackage{mathrsfs}
\usepackage{amsmath,amssymb, amsfonts, amsthm, latexsym,titlesec,tikz}

\newtheorem{lem}{Lemma}[section]
\newtheorem{lemma}[lem]{Lemma}
\newtheorem{cor}[lem]{Corollary}
\newtheorem{propos}[lem]{Proposition}
\newtheorem{thm}[lem]{Theorem}
\newtheorem{rem}[lem]{Remark}
\newtheorem{defin}[lem]{Definition}

\newcommand{\comment}[1]{}

\newcommand{\N}{\mathbb{N}}

\newcommand{\Z}{\mathbb{Z}}

\newcommand{\R}{\mathbb{R}}
\newcommand{\Ex}[1]{\mathbb E\left[#1\right]}

\newcommand{\Tr}[1]{\textnormal{Tr}\left(#1\right)}
\newcommand{\Var}[1]{\textnormal{Var}\left(#1\right)}
\newcommand{\Cov}[1]{\textnormal{Cov}\left(#1\right)}
\numberwithin{equation}{section}

\begin{document}

\title{Limit Theorems on the Mesoscopic Scale for the Anderson Model}
\author{Yoel Grinshpon \thanks{Institute of Mathematics, The Hebrew University of Jerusalem, Jerusalem, 91904, Israel.
Email: yoel.grinshpon@mail.huji.ac.il}}

\maketitle

\begin{abstract}
In this paper, we study eigenvalue fluctuations of the finite volume Anderson model in the mesoscopic scale. We carry out this study in a regime of exponential localization and prove a central limit theorem for the eigenvalue counting function in a shrinking interval.
\end{abstract}

%%%%%%%%%%%%%%%%%%%SECTION 1%%%%%%%%%%%%%%%%%%%%%%%%%
\section{Introduction and Preliminaries}
The purpose of this paper is to prove a mesoscopic central limit theorem (CLT) for the eigenvalue counting function of finite truncations of the Anderson model in a regime of energies where localization holds.

The discrete Anderson model on $\Z^d$ is the random operator
$$H:\ell^2(\Z^d)\longrightarrow \ell^2(\Z^d)$$
$$H = \Delta + V $$
where $\Delta$ is the discrete Laplacian and $V$ is a multiplication operator, i.e.,
$$(H u)_n = \sum_{m\sim n}u_m +V_n \cdot u_n$$
where $V_n$  $(n \in \Z^d )$ are i.i.d.\ random variables. We will assume that the distribution of $V_n$ is absolutely continuous with respect to the Lebesgue measure with density $\rho(v)\textnormal{d}v$ satisfying $||\rho||_\infty < \infty$.

We will be looking at finite truncations of $H$,
$$H_L = \chi_L^d H \chi_L^d $$
as $L\longrightarrow \infty$ where $\chi_L^d$ is the indicator function on the cube
$$\Lambda_L = [-L,L]^d\cap \Z^d,$$
and we denote by
$$E_1^L \le ... \le E_{|\Lambda_L|}^L $$
the eigenvalues of $H_L$.
The \emph{empirical measure} of $H_L$ is the measure
$$\textrm{d}\mu_L=\frac{1}{|\Lambda_L|}\sum_{i=1}^{|\Lambda_L|} \delta_{E_i^L}$$
where $\delta_{E_i^L}$ is the Dirac measure at $E_i^L$. When the empirical measure has a limit as $L \rightarrow \infty$, this limit $\mu$ is known as the density of states (DOS) of $H$. In our case, this limit indeed exists and is known to be absolutely continuous with respect to the Lebesgue measure almost everywhere \cite[Ch.\ 4]{AW}, with a Radon-Nikodym derivative
$$\textrm{d}\mu(E)=f(E)\textrm{d}E.$$

 Our goal will be to understand the fluctuations of $\mu_L$ on the mesoscopic scale and establish the convergence of these fluctuations to a Gaussian limit. In other words, we shall look at intervals of length $\sim \frac{1}{|\Lambda_L|^\eta}$ around a certain energy $E$ where $0 < \eta <1$, and study the fluctuations of the counting function of the eigenvalues in these intervals as $L$ tends to $\infty$.
We study these fluctuations in a regime of localization, i.e., where the spectrum of $H$ is pure-point with exponentially decaying eigenfunctions.

We define the Green's function  $G_{\Lambda}(x,x;z)$ on some box $\Lambda \subseteq \Z^d$ to be
$$G_\Lambda(x,y;z)=\langle x,\left(H_{\Lambda}-z\right)^{-1}y \rangle$$
(where $H_\Lambda$ is $H$ restricted to $\Lambda$), and the set $\mathcal{L}$ to be the following:

\begin{defin}
We say that $E \in \mathcal{L}$, if
\begin{enumerate}
  \item $E \in \sigma(H)$.
  \item $f(E)$ exists and is positive.
  \item There exist $s \in (0,1)$, $C_1>0$, $C_2$ and $r>0$ such that
  $$\Ex{\left|G_\Lambda(x,y;z)\right|^s} \le C_1e^{-C_2|x-y|} $$
for any hypercube $\Lambda \subset \Z^d$, $x \in \Lambda$, $y \in \partial \Lambda$ and $z \in \mathbb{C}_+$ such that $|z-E|<r \ 
\big($we say that $y \in \partial \Lambda$ if $y \in \Lambda$ and there is a $y' \in \partial \Lambda$ such that $|y-y'| \big)$.
\end{enumerate}

\end{defin}

Limit theorems for the fluctuations of the eigenvalues in the Anderson model have already been obtained, mostly on the microscopic scale ($\eta = 1$) and on the macroscopic scale ($\eta = 0$). Minami \cite{Minami} proved that under a certain assumption which implies localization ($E \in \mathcal{L}$, an assumption which we shall adopt as well), the eigenvalue point process converges to a Poisson point process on the microscopic scale. This was after Molchanov \cite{M} did so for the continuous case. This result can be interpreted as statistical independence of the eigenvalues as $L \longrightarrow \infty $ in a window of length $\sim \frac{1}{|\Lambda_L|}$ around some energy $E \in \sigma(H)$.  

 Several macroscopic limit theorems have been proven for the Anderson model \cite{KP,PS,Reznikova} in one-dimension, all showing Gaussian behavior of the trace of $f(H_L)$ for different functions $f$ with a variance that grows proportionally to $L$. We remind the reader that in one-dimension, there is always localization for the Anderson model \cite[Ch.\ 9]{CFKS}. In addition, the author and White proved a macroscopic CLT for polynomials of the multi-dimensional Anderson model \cite{GW}. Indeed, for $p(H_L)$, if the distribution of the elements of the potential $V$ is supported on more than three points, the variance of $\Tr{p(H_L)}$ is of magnitude $L^d$.
  These results (or more specifically, the growth rate of the variance) on the macroscopic scale resemble the expected behavior of a sum of i.i.d.\ random variables, and therefore can be interpreted as some sort of limit independence of the eigenvalues of $H_L$ as $L$ tends to $\infty$. A natural question now arises - in regions where localization holds, can this limit independence be seen in the scales between the microscopic and the macroscopic, i.e., in windows of magnitude $\sim \frac{1}{|\Lambda_L|^\eta}$ around some energy $E$?

Some results have already been achieved for the mesoscopic scale. Germinet and Klopp \cite{KG} established several results regarding the level spacing statistics and the localization centers of the eigenvalues of many Schr\"odinger operators (including the Anderson model) where localization holds, both on the microscopic scale and on mesoscopic scales.
Warzel and Von Soosten \cite{VSW} proved a law of large numbers for the eigenvalues encompassing a wide range of Schr\"odinger operators, which include the Anderson model in one-dimension on some mesoscopic scales. It is worth pointing out that our method of proof utilizes Minami's result along with general features of the Anderson model, and does not require additional hard analysis.

Mesoscopic scale fluctuations of eigenvalues are of great interest in the field of random matrix theory. Such theorems have been proved for the CUE \cite{Sosh} and the GOE \cite{BMK}. In recent years there has been a plethora of such results. \cite{BD}\cite{FKS}\cite{DJ}\cite{Lam} is a very partial list of examples for such results.

For $a,b \in \R$ such that $a < 0 <b$ and $E \in \sigma(H)$, we shall define $X_L = X_L(\eta,E,a,b)$ to be the number of eigenvalues of $H_L$ in the interval
$$I_L = I_L\left(E,L,\eta,a,b\right) \equiv \left(E + \frac{a}{|\Lambda_L|^\eta}, E + \frac{b}{|\Lambda_L|^\eta} \right). $$
We will prove the following theorems:

\begin{thm} \label{thm:wlln}
For $a,b \in \R$ such that $a < 0 < b $, $0 < \eta < 1$ and $E \in \mathcal{L}$,
$$\frac{X_L}{|\Lambda_L|^{1-\eta}}\longrightarrow f(E)\cdot (b-a)$$
in probability.
\end{thm}

\begin{thm} \label{thm:main}

For $a,b \in \R$ such that $a < 0 < b $, $0 < \eta < 1$ and $E \in \mathcal{L}$,
$$\frac{X_L-|\Lambda_L|^{1-\eta} \cdot f(E)\cdot(b-a)}{\sqrt{|\Lambda_L|^{1-\eta}}} \overset{d}{\longrightarrow} N(0,\sigma^2) $$
where $N(0,\sigma^2)$ is a Gaussian random variable with mean $0$ and variance
$$\sigma^2 = f(E)^2\cdot (b-a)^2. $$
\end{thm}

\textbf{Acknowledgments.} The author wants to express his deep gratitude to Jonathan Breuer for his thorough and committed guidance throughout this research. The author would also like to thank Daniel Ofner, Asaf Shachar and Eyal Seelig for useful discussions. 

This research was supported in part by the Israel Science Foundation (Grant No. 1378/20).

\section{Overview and Preliminaries}

In this paper, we analyze the eigenvalue fluctuations on the mesoscopic scale using Minami's result for the microscopic scale. However, we will need a slightly modified version of Minami's theorem:

\begin{propos} \label{thm:minami}
Let $H$ be the Anderson model on $\Z^d$ with a potential $V$ with distribution $\rho (v) \textrm{d}v$ such that $||\rho||_\infty < \infty$.
For every $L$, $a_L \in \R^d$ and $c_L\ge 0$, we define the boxes
$$\Lambda_{L,a_L,c_L} = \left([-L+c_L,L-c_L]^d+a_L\right) \cap \Z^d$$
and denote by $E_j^L$ the eigenvalues of $H$ restricted to $\Lambda_{L,a_L,c_L}$ (denoted by $H_{\Lambda_L,a_L,c_L}$).
Assuming that $c_L \underset{L \rightarrow \infty}{\longrightarrow} 0 $, for every $E \in \mathcal{L}$, the process
$$\mu_{L,a_L,c_L}^E = \sum_{j=1}^{\left|\Lambda_{L,a_L,c_L}\right|} \delta_{\left|\Lambda_{L,a_L,c_L}\right|\cdot(E_j^L-E)}$$
converges to a Poisson point process with intensity
$$\frac{\textnormal{d}\nu}{\textnormal{d}E}=f(E)$$
where $f$ is the Radon-Nikodym derivative of the DOS with respect to the Lebesgue measure.

\end{propos}

The next Corollary immediately follows.
\begin{cor} \label{cor:minami}
Under the assumptions of Theorem \ref{thm:minami}, given $a,b \in \R$ such that $a < 0 <b$, the number of eigenvalues of $\Lambda_{L,a_L,c_L}$ in the interval
$$I=\left(E+\frac{a}{(2L+1)^d},E+\frac{b}{(2L+1)^d}\right) $$
(denoted by $Z_L^I=Z_L^I(E,a,b,a_L,c_L)$) converges in distribution to a Poisson random variable with parameter $$\lambda = f(E)\cdot(b-a).$$
\end{cor}

The proof of Proposition \ref{thm:minami} is essentially identical to the proof of Minami's original result \cite{Minami} so we omit it.

We present here a brief overview of the proof of Theorem \ref{thm:main}. The general goal of our proof is to understand the fluctuations of the eigenvalues on the mesoscopic scale exploiting our knowledge of the fluctuations on the microscopic scale given in Corollary \ref{cor:minami}. We would like to do this in the following manner. We create a partition of $\Lambda_L$ to smaller boxes $\Lambda_{L,j}$ which have side lengths of magnitude $L^\eta$. Restricting $H$ to the new boxes $\Lambda_{L,j}$ (denoting it $H_{L,j}$ accordingly), from Corollary \ref{cor:minami}, the number of eigenvalues of $H_{L,j}$ in the intervals $I_L$ converges to a Poisson random variable. Since the eigenfunctions are exponentially localized, the transition from $H_L$ to $H_{L,j}$ will not significantly affect most eigenvalues of $H_L$. For relatively large $\eta$, this approximation indeed works (see Proposition \ref{prop:mainprop}), and one can infer a CLT for the eigenvales of $H_L$ in $I_L$ as a sum of independent random variables which converge to a Poisson random variable (Proposition \ref{prop:localclt}). However, for relatively small $\eta$, this approximation does not work. In this case, we use an inductive approach. This will be done in the proof of Theorem \ref{thm:main}.

In our proofs, convergence of the variance of $Z_L^I$ (as defined in Corollary \ref{cor:minami}) to the variance of a Poisson random variable would do us a great service. In general, convergence in distribution does not imply convergence of the variances, but it is true in our case.

\begin{lemma} \label{lemmma:moments}
Under the assumptions of Proposition \ref{thm:minami}, the $k$-th moment of $Z_L^I$ (as defined in Corollary \ref{cor:minami}) converges to the $k$-th moment of a Poisson random variable with parameter
$$\lambda =  f(E)\cdot (b-a). $$
In particular,
$$\Ex{Z_L^I} \overset{d}{\longrightarrow} \lambda $$
and
$$\Var{Z_L^I} \longrightarrow \lambda^2.$$
\end{lemma}
In order to prove Lemma \ref{lemmma:moments}, we will need Theorems \ref{thm:moments} and \ref{minamiestimate}:

\begin{thm}\textnormal{\cite[Corollary from Theorem 25.12]{Billingsley}} \label{thm:moments}
Let r be a positive integer and $\epsilon >0 $. If $X_n \longrightarrow X$ in distribution and
$$\underset{n}{\sup \ } \Ex{|X_n|^{r+\epsilon}}< \infty. $$
Then $\Ex{|X|^r}< \infty$, and
$$\Ex{X_n^r}\longrightarrow \Ex{X^r}. $$
\end{thm}

\begin{thm} \textnormal{\cite[Corollary 2.4]{CGK}} \label{minamiestimate}
Let 
$$H = \Delta+V $$
be the Anderson model on a finite graph $G$ such that for every $x \in G$, the distribution of $V_x$ is a.c.\ with respect to the Lebesgue measure with a Radon-Nikodym derivative $\rho$, and for any interval $I \subset \R$, define $Z^I$ to be the number of eigenvalues of $H$ in $I$.  Then for every interval $I \subset \R$, there exists a constant $C>0$ which depends only on the distribution $\rho$ such that
$$P\left(Z^I \ge n \right) \le \frac{C^n\cdot |G|^n\cdot |I|^n}{n!}.$$
\end{thm}

\begin{proof} [Proof of Lemma \ref{lemmma:moments}]
By Theorem \ref{thm:moments} it is enough to show that
$$\underset{L}{\sup \ }  \Ex{|Z_L^I|^k}< \infty. $$
for every $k \in \N$.
Note that $Z_L^I$ is a discrete random variable which takes values in $\N\cup\{0\}$ and hence from Theorem \ref{minamiestimate}, there exists a constant $C>0$ such that for any $L>0$,
$$\Ex{|Z_L^I|^k} = \sum_{n=1}^\infty n^k \cdot P\left(Z_L^I = n \right) \le  \sum_{n=1}^\infty n^k \cdot P\left(Z_L^I \ge n \right)  $$
$$\le \sum_{n=1}^\infty \frac{C^n}{n!} \cdot n^k \cdot |\Lambda_{L,a_L,c_L}|\cdot \frac{1}{|2L|^d} < \infty.$$

\end{proof}

\begin{rem} \label{rem:uniform_convergence}
The convergence of $\Var{Z_L^I}$ to the variance of a Poisson random variable with parameter $\lambda$  does not depend on $a_L,c_L$, in a sense that for a given $\epsilon > 0$, there exists $L_0 > 0 $ that such that if $L>L_0$,
$$\left|\Var{Z_L^I}-\lambda^2 \right| <\epsilon $$
for every $a_L \in \R^d$, $c_L < 1$.
Define
$$D = \{0,1\}^d $$
and corresponding random variables
$$\{W_L^k\}_{k\in D}$$
where $W_L^k$ is the number of eigenvalues of $H$ restricted to the box
$$\left(\left(0,2L+k_1\right)\times \left(0,2L+k_2\right)\times...\times \left(0,2L+k_d\right)\right)\cap \Z^d$$
in the interval
$$\left(\frac{a}{(2L+1)^d},\frac{b}{(2L+1)^d}\right). $$
For every $L\in \N$ $a_L\in \R^d$, and $0 \le c_L <1$, there exists $k \in D$ such that $Z_L^I$ and $W_L^k$ share the same distribution. From Corollary \ref{cor:minami} and Lemma \ref{lemmma:moments}, for a fixed $k \in D$, the variance of $W_L^k$ converges to $\lambda^2$ as $L \longrightarrow \infty$. Since $\{W_L^k\}_{k\in D}$ is a finite set, so does $\Var{Z_L^I}$.

This type of argument will appear several times throughout this paper.

\end{rem}

\section{The Proofs}
Similarly to the proof strategy in \cite{Minami}, for every $L$ and $0 < \beta <1$ we shall divide $\Lambda_L$ into separate $M_L(\beta)$ boxes $\Lambda_{L,j}$ in the following manner. We start by forming a partition of each edge of $\Lambda_L$ into $\left\lceil (2L)^{1-\beta} \right\rceil$ intervals of length
$$\frac{2L}{\left\lceil (2L)^{1-\beta} \right\rceil}.$$
This forms a partition of $\Lambda_L$ into $M_L(\beta)=\left(\left\lceil (2L)^{1-\beta} \right\rceil\right)^d$ boxes $\Lambda_{L,j}$ and induces the corresponding measures
$$\mu_{L,j}^E \equiv \mu_{L,j,\eta,\beta}^E = \sum_{i=1}^{|\Lambda_{L,j}|} \delta_{|\Lambda_{L,j}|^\eta \cdot (E_{i,j}^L-E)} $$
where $E_{i,j}^L$ is the $i$-th eigenvalue of $H_{L,j}$, the operator $H$ restricted to the $j$-th box.
Our goal will be to approximate $\mu_L^E$ with $\mu_{L,j}^E$ in the sense that
$$\underset{L \longrightarrow \infty}{\lim}\Ex{\frac{1}{|\Lambda_L|^\alpha} \left(\mu_L^E(g)-\mu_{L,j}^E(g)\right)}=0 $$
for a suitable $\alpha>0$ and certain functions $g \in L^1(\R)$.

\begin{rem}
As mentioned, for every box $\Lambda_{L,j}$, the edges of $\Lambda_{L,j}$ are of length $2L' = \frac{2L}{\left\lceil (2L)^{1-\beta}\right\rceil}$ and for every $\epsilon > 0$, for $L$ large enough,
$$(2L)^\beta - \epsilon \le  2L' \le (2L)^\beta, $$
so for some $L_0 \in \R$ and every $L>L_0$, there exist $c_L>0$ such that $c_L\rightarrow 0$ and $a_L \in \R^d$ such that
$$\Lambda_{L,j} = \left([-L^\beta+c_L,L^\beta-c_L]^d+a_L\right)\cap \Z^d $$
and therefore Corollary \ref{cor:minami} is applicable for the boxes $\Lambda_{L,j}$ in the sense that since the side lengths of $\Lambda_{L,j}$ converge to $(2L)^\beta$, the number of eigenvalues of $H_{L,j}$ in $I_{L^\beta}$ converges to a Poisson random variable.
\end{rem}

For the proof of the next proposition, we will need the following lemma.

\begin{lemma} \label{lemma:internal}
Let $E \in \mathcal{L}$. Then there exist $s \in (0,\frac12)$ and $B_1,B_2,r>0$ such that for any $\Lambda \subseteq \Lambda ' \subseteq \Z^d$ (where $\Lambda$ is finite), $x \in \Lambda$ and $z \in \mathbb{C}_+$ such that $|z-E|<r$,

$$\Ex{|G_\Lambda(x,x;z)-G_{\Lambda'}(x,x;z)|}\le \frac{B_1}{Im z^{2(1-s)}}e^{-B_2 dist(x,\partial \Lambda)}. $$
\end{lemma}

The proof is very similar to the proof of Lemma 17.9 in \cite{AW} so we omit it.

\begin{propos} \label{prop:mainprop}
Denote by $g \equiv \chi_{(a,b)}$ the indicator function of the interval $(a,b)$ for $a < 0 <b$. For $\alpha,\beta,\eta>0$ such that $\alpha+\eta>1-\beta$ and $E\in \mathcal{L}$,
$$\underset{L \rightarrow \infty}{\lim} \frac{1}{|\Lambda_L|^{\alpha}} \Ex{\left|\mu_L^E(g)-\sum_{j=1}^{M_L(\beta)} \mu_{L,j}^E(g)\right|}=0. $$
\end{propos}

In order to prove Proposition \ref{prop:mainprop}, we begin by proving the analogous statement for a different set of functions. We define for every $z\in \mathbb{C}_+$,
$$\phi_z(u)=\frac{1}{\pi}Im\frac{1}{u-z}. $$
We shall prove the following lemma.
\begin{lemma} \label{lemma:mainprop}
For any $z\in \mathbb{C}_+$, $\alpha,\beta,\eta>0$ such that $\alpha+\eta>1-\beta$ and $E\in \mathcal{L}$,
$$\underset{L \rightarrow \infty}{\lim} \frac{1}{|\Lambda_L|^{\alpha}} \Ex{\left|\mu_L^E(\phi_z)-\sum_{j=1}^{M_L(\beta)} \mu_{L,j}^E(\phi_z)\right|}=0. $$
Moreover, taking $z  = \frac{i}{|\Lambda_L|^2}$, the statement still holds as $L$ tends to $\infty$, i.e.,
$$\underset{L \rightarrow \infty}{\lim} \frac{1}{|\Lambda_L|^{\alpha}} \Ex{\left|\mu_L^E\left(\phi_{\frac{i}{|\Lambda_L|^2}}\right)-\sum_{j=1}^{M_L(\beta)} \mu_{L,j}^E\left(\phi_{\frac{i}{|\Lambda_L|^2}}\right)\right|}=0. $$
\end{lemma}

\begin{proof} [Proof of Lemma \ref{lemma:mainprop}]
This proof is quite similar to Step 3 in \cite{Minami} and Lemma 17.7 in \cite{AW}. In order to proceed, we shall choose some $c>0$, and break each $\Lambda_{L,j}$ into separate components:
$$\Lambda_{L,j}^i = \{x \in \Lambda_{L,j} \ : \ d(x,\partial \Lambda_{L,j}) > L^c\}$$
$$\Lambda_{L,j}^b=\Lambda_{L,j}-\Lambda_{L,j}^i. $$
In other words, $\Lambda_{L,j}^b$ is the boundary of $\Lambda_{L,j}$ and $\Lambda_{L,j}^i$ is the interior of $\Lambda_{L,j}$.
Denoting $z_L=\frac{z}{|\Lambda_L|^\eta}$,
\begin{equation}
\begin{split}
  & \ \ \ \ \frac{1}{|\Lambda_L|^{\alpha}}  \left(\mu_L^E(\phi_z)-\sum_{j=1}^{M_L(\beta)} \mu_{L,j}^E(\phi_z)\right) \\
  &=\frac{1}{\pi} \frac{1}{|\Lambda_L|^{\alpha}} \sum_{i=1}^{|\Lambda_L|} Im \left(\frac{1}{|\Lambda_L|^\eta (E_i^L-E)-z} \right) \\
  &-\frac{1}{\pi |\Lambda_L|^\alpha} \sum_{j=1}^{M_L(\beta)}\sum_{i=1}^{|\Lambda_{L,j}|} Im \left(\frac{1}{|\Lambda_L|^\eta (E^L_{i,j}-E)-z} \right) \\
  &= \frac{1}{\pi |\Lambda_L|^{\alpha+\eta}}\left(\Tr{Im\left(H_{\Lambda_L}-z_L \right)^{-1}}-\sum_{j=1}^{M_L(\beta)} \Tr{Im\left(H_{\Lambda_{L,j}}-z_L \right)^{-1}} \right) 
\end{split}
\end{equation}
\begin{equation}
\begin{split}
  &=\frac{1}{\pi |\Lambda_L|^{\alpha+\eta}} \sum_{j=1}^{M_L(\beta)} \left( \sum_{x \in \Lambda_{L,j}} \left(Im \ G_{\Lambda_L} (x,x;z_L)-Im  \ G_{\Lambda_{L,j}}(x,x;z_L)\right)  \right)  \\
  &=\frac{1}{\pi |\Lambda_L|^{\alpha+\eta}} \sum_{j=1}^{M_L(\beta)}  \left(\sum_{x \in \Lambda^b_{L,j}} \left(Im \ G_{\Lambda_L} (x,x;z_L)-Im  \ G_{\Lambda_{L,j}}(x,x;z_L) \right) \right) \\
  &+\frac{1}{\pi |\Lambda_L|^{\alpha+\eta}} \sum_{j=1}^{M_L(\beta)}  \left(\sum_{x \in \Lambda^i_{L,j}} \left( Im \ G_{\Lambda_L} (x,x;z_L)-Im  \ G_{\Lambda_{L,j}}(x,x;z_L) \right) \right).
\end{split}
\end{equation}
Now we shall look at the expectation of the RHS and examine each sum separately.
$$\frac{1}{\pi |\Lambda_L|^{\alpha+\eta}} \Ex{\sum_{j=1}^{M_L(\beta)}  \left(\sum_{x \in \Lambda^b_{L,j}} Im \ G_{\Lambda_L} (x,x;z_L)-Im  \ G_{\Lambda_{L,j}}(x,x;z_L) \right)}  $$
$$\le \frac{1}{\pi |\Lambda_L|^{\alpha+\eta}} \sum_{j=1}^{M_L(\beta)}  \sum_{x \in \Lambda^b_{L,j}} \left(\Ex{ Im \ G_{\Lambda_L} (x,x;z_L)}+\Ex{\left(Im  \ G_{\Lambda_{L,j}}(x,x;z_L) \right)}\right) $$
$$\le \frac{1}{\pi |\Lambda_L|^{\alpha+\eta}} \sum_{j=1}^{M_L(\beta)} \sum_{x \in \Lambda^b_{L,j}} 2 ||\rho||_\infty $$
$$= \frac{M_L(\beta)}{\pi |\Lambda_L|^{\alpha+\eta}}\cdot |\Lambda^b_{L,j}|\cdot 2||\rho||_\infty,  $$
where the last inequality is true due to the fact that
$$|\Ex{Im\left(G_\Lambda(x,x;z)\right)}|\le \pi ||\rho||_\infty $$
for any $\Lambda \subset \Z^d$, $x\in \Lambda$, $z\in \mathbb{C}_+$ (see equations (2.19)--(2.22) in \cite{Minami}). Notice that this bound does not depend on $z$, so as long as we take $c$ small enough and under our assumptions on $\alpha$ and $\beta$,
%(such that $c\cdot d< \alpha+\beta+\eta-1$),
$$\frac{1}{\pi |\Lambda_L|^{\alpha+\eta}} \Ex{\sum_{j=1}^{M_L(\beta)}  \left(\sum_{x \in \Lambda^b_{L,j}} Im \ G_{\Lambda_L} (x,x;z_L)-Im  \ G_{\Lambda_{L,j}}(x,x;z_L) \right)} \underset{L\longrightarrow \infty}{\longrightarrow} 0. $$

As for the second sum, since $z \in \mathcal{L}$ and since for any $r$, $|z_L-E|<r$ for $L$ large enough, from Lemma \ref{lemma:internal}, we obtain
$$\Ex{\sum_{x \in \Lambda^i_{L,j}} Im \ G_{\Lambda_L} (x,x;z_L)-Im  \ G_{\Lambda_{L,j}}(x,x;z_L)} $$
$$\le \frac{M_L(\beta) \cdot |\Lambda_{L,j}|}{|\Lambda_L|^{\alpha+\eta}}\frac{B_1 \cdot |\Lambda|^{1-\beta}}{Im z^{2(1-s)}} e^{-B_2 dist(\Lambda^i_{L,j},\partial \Lambda_{L,j})}\underset{L\longrightarrow \infty}{\longrightarrow} 0 $$
as long as $Im z$ decays polynomially in $L$.
Hence, we obtain the desired result.

\end{proof}

Having established Lemma \ref{lemma:mainprop}, we can now prove Proposition \ref{prop:mainprop}.

\begin{proof} [Proof of Proposition \ref{prop:mainprop}]
We shall use the set of functions $\phi_z$ with $z=i\epsilon$, $\epsilon > 0 $ to approximate the indicator functions $g \equiv \chi_{(a,b)}$.
For every $\epsilon > 0$, define
$$g_\epsilon = \phi_{i\epsilon}*g $$
and
$$\nu_L = \frac{1}{|\Lambda_L|^\alpha}\left(\mu_L^E-\sum_{j=1}^{M_L}\mu_{L,j}^E\right) $$
We will show that for $\epsilon = \frac{1}{|\Lambda_L|^2}$,
\begin{enumerate}
  \item $$\underset{L \rightarrow \infty}{\lim}  \Ex{\left|\nu_L(g_\epsilon) \right|} \overset{L \rightarrow \infty}{\longrightarrow} 0 $$

  \item $$\underset{L \rightarrow \infty}{\lim}  \Ex{\left|\nu_L(g-g_\epsilon) \right|}\overset{L \rightarrow \infty}{\longrightarrow} 0 $$

\end{enumerate}
and thus obtain the desired result.
For (1),
\begin{equation} \label{eq:dominated}
\Ex{\left|\nu_L(g_\epsilon)\right|}\le \int_{-\infty}^\infty |g(y)| \Ex{\nu(\phi_{i\epsilon+y})}dy=\int_a^b \Ex{\nu(\phi_{i\epsilon+y})}dy
%\overset{L \rightarrow \infty}{\longrightarrow} 0.
\end{equation}
For $L$ large enough, $\left|\frac{y+i\epsilon}{|\Lambda_L|^\eta}\right|< r$, so according to Lemma \ref{lemma:mainprop},
$$\Ex{\nu(\phi_{i\epsilon+y})} \underset{L \rightarrow \infty}{\longrightarrow} 0 $$
uniformly, and by the Dominated Convergence Theorem, the RHS of \ref{eq:dominated} converges to $0$ as well.

For (2), since $0 \le g_\epsilon (x) \le 1$ for any $x \in \R$ (denoting by $||\cdot ||_1$ the $L^1$ norm on $\R$),
$$\Ex{|\nu_L(g-g_\epsilon)|} \le \Ex{\mu_L^E(g-g_\epsilon)}+\sum_{j=1}^{M_L\left(\beta\right)}\Ex{\mu_{L,j}^E(g-g_\epsilon)}$$
$$\le 2 |\Lambda_L|\cdot ||g-g_\epsilon||_1 \overset{L \rightarrow \infty}{\longrightarrow} 0$$
if
$$||g-g_\epsilon ||_1 = o\left(|\Lambda_L|^{-1.5}\right), $$
which is the statement of Lemma \ref{lemma:mollifier}.

\end{proof}

\begin{lemma} \label{lemma:mollifier}
For $\epsilon = \frac{1}{|\Lambda_L|^2}$,
$$||g-g_\epsilon ||_1 = o\left(|\Lambda_L|^{-1.5}\right) $$
\end{lemma}

\begin{proof}
Using direct integration, one can see that

$$g_\epsilon(x) = \frac{1}{\pi} \left(\arctan\left(\frac{x-a}{\epsilon}\right)- \arctan\left(\frac{x-b}{\epsilon}\right) \right) $$
and therefore, denoting
$$A_\epsilon(x) = \frac{1}{\pi} \left((x-a)\arctan\left(\frac{x-a}{\epsilon}\right)- (x-b)\left(\arctan\left(\frac{x-b}{\epsilon} \right)\right) \right)$$
$$B_\epsilon(x) = \frac{1}{2\pi}\cdot \epsilon\left(\log\left(1+\left(\frac{x-a}{\epsilon}\right)^2\right)-\log\left(1+\left(\frac{x-b}{\epsilon}\right)^2\right)\right), $$
we get
$$\int g_\epsilon \textnormal{d}x = A_\epsilon(x) - B_\epsilon(x).$$

Using a direct calculation, one can verify that
$$A_\epsilon(a)- \underset{x \rightarrow - \infty}{\lim} A_\epsilon (x)=O(\epsilon), \underset{x \rightarrow -\infty}{\lim} B_\epsilon(x) - B_\epsilon(a)=O(\epsilon^\alpha)$$ for any $\alpha < 1$, so
$$\int_{-\infty}^a g_\epsilon \textnormal{d}x = O(\epsilon^\alpha) $$
for any $\alpha < 1$. A similar argument shows the same for
$$\int_b^\infty g_\epsilon \textnormal{d}x .$$
Finally,
$$(b-a)-A_\epsilon(b)+A_\epsilon(a) = O(\epsilon), \  B_\epsilon(b)-B_\epsilon(a) = O(\epsilon^\alpha)$$
so
$$\int_a^b \left(1-g_\epsilon\right) \textnormal{d}x = O(\epsilon^\alpha)$$
for any $0< \alpha <1$ as well.
\end{proof}

From Proposition \ref{prop:mainprop}, we can easily derive Theorem \ref{thm:wlln}:

\begin{proof} [Proof of Theorem \ref{thm:wlln}]
Take $\beta = \eta$, and $\alpha = 1-\eta$. Notice that for our choice of $\alpha$ and $\beta$, the assumptions of Proposition \ref{prop:mainprop} hold for any $0<\eta<1$. Hence, taking
%$$I = \left(a,b\right), $$
$$X_L = \mu_L^E(\chi_{(a,b)}), \ X_{L,j} = \mu_{L,j}^E(\chi_{(a,b)}) $$
and applying Proposition \ref{prop:mainprop},
%to $$g = \chi_{(a,b)}, $$
$$\frac{X_L-\sum_{j=1}^{M_L(\beta)} X_{L,j}}{|\Lambda_L|^{1-\eta}}\underset{L\longrightarrow \infty}{\longrightarrow} 0 $$
in probability.
From Proposition \ref{thm:minami}, $X_{L,j}\underset{L \longrightarrow \infty}{\longrightarrow} Pois(\lambda)$ with 
$$\lambda = f(E)\cdot (b-a)$$
(see Corollary \ref{cor:minami}). Together with Lemma \ref{lemmma:moments} and Remark \ref{rem:uniform_convergence}, this implies that
$$\Ex{X_{L,j}} \longrightarrow \lambda  $$
$$\Var{X_{L,j}} \longrightarrow \lambda^2 $$
uniformly in $j$.
 This means that the variance of $X_{L,j}$ is uniformly bounded. In addition, for each $L$, $X_{L,j}$ are independent, which altogether implies
  $$\Var{\frac{\sum_{j=1}^{M_L\left(\beta\right)}X_{L,j}}{M_L\left(\beta\right)}} = \frac{1}{M_L\left(\beta\right)^2}\sum_{j=1}^{M_L\left(\beta\right)} \Var{X_{L,j}} \longrightarrow 0 $$
  $$\Ex{\frac{\sum_{j=1}^{M_L\left(\beta\right)}X_{L,j}}{M_L\left(\beta\right)}} = \frac{1}{M_L\left(\beta\right)}\sum_{j=1}^{M_L\left(\beta\right)} \Ex{X_{L,j}} \longrightarrow f(E)\cdot(b-a), $$
and hence
$$\frac{\sum_{j=1}^{M_L(\beta)} X_{L,j}-M_L(\beta)\cdot f(E) \cdot (b-a)}{M_L(\beta)}\underset{L\rightarrow \infty}{\longrightarrow} 0 $$
in probability. Since
$$\frac{M_L(\beta)}{\left| \Lambda_L \right|^{1-\eta}}\longrightarrow 1, $$
we obtain the desired result.
\end{proof}

After obtaining the law of large numbers, we move on to our CLT. First, we begin with a relatively simple case where $\eta$ is large.

\begin{propos} \label{prop:localclt}
For $a_L \in \R^d$ and $c_L \ge 0 $, define
$$\Lambda_L \equiv \Lambda_{L,a_L,c_L} = \left([-L+c_L,L-c_L]+a_L\right)\cap \Z^d.$$
For $E \in \mathcal{L}$, $\tilde{\eta}>\frac12$ $a,b \in \R$ such that $a<0<b$ and $c_L \ge 0$ such that $c_L \underset{L \rightarrow \infty}{\longrightarrow} 0$,
$$\frac{X_L-|\Lambda_L|^{1-\tilde{\eta}}\cdot f(E)\cdot (b-a)}{\sqrt{|\Lambda_L|^{1-\tilde{\eta}}}}\underset{L\longrightarrow \infty}{\overset{d}{\longrightarrow}} N(0,\sigma^2) $$
where $\sigma^2 = f(E)^2\cdot (b-a)^2$.
Moreover,
$$\Var{\frac{X_L}{\sqrt{|\Lambda_L|^{1-\tilde{\eta}}}}}\longrightarrow \sigma^2 $$
as well.
\end{propos}

In order to prove Proposition \ref{prop:localclt} and Theorem \ref{thm:main}, we will need the Lindeberg-Feller CLT for arrays.

\begin{thm}\textnormal{\cite[Theorem 3.4.5]{Durett}} \label{thm:Lindeberg}
For each $n\in \N$, let $X_{n,m}$ $1 \le m \le n $ be independent random variables with $\Ex{X_{n,m}}=0$. Suppose
\begin{enumerate}
  \item $$\sum_{m=1}^n \Ex{X_{n,m}^2}\overset{n \rightarrow \infty}{\longrightarrow} \sigma^2 >0 $$
  \item For all $\epsilon > 0$,
  $$\underset{n \rightarrow \infty}{\lim} \Ex{X_{n,m}^2\cdot \chi_{|X_{n,m}|> \epsilon}}=0. $$
\end{enumerate}
Then
$$\frac{X_{n,1}+...+X_{n,n}}{\sqrt{\Var{\sum_{m=1}^n X_{n,m}}}} \overset{d}{\longrightarrow} N(0,1)  $$
as $n\longrightarrow \infty$.
\end{thm}

\begin{proof} [Proof of Proposition \ref{prop:localclt}]
We start by creating a partition of $\Lambda_L$ (as described in the beginning of Section 3) with $\beta=\tilde{\eta}$ into separate boxes $\{\Lambda_{L,j}\}_{j=1}^{M_L\left(\beta\right)}$ (see Figure 1) and for every $\Lambda_{L,j}$, we define $X_{L,j}$ to be the number of eigenvalues of $H_{\Lambda_{L,j}}$ in $I_L$. From Corollary \ref{cor:minami}, for every fixed $1\le j \le M_L\left(\beta\right)$
$$X_{L,j}\overset{d}{\longrightarrow} Pois(\lambda) $$
with $\lambda = f(E)\cdot (b-a)$.

\begin{figure} \label{fig:bigeta}
\centering
\begin{tikzpicture}
\draw[thick] (0,0) rectangle (5,5);
\node[text width=6cm, anchor=west, right] at (2,6.5)   {$\Lambda_{L,j}$};
\draw[thick,->] (2.5,6.3) -- (2.5,5.5);

\draw[thick,<->] (-0.25,5) -- (-0.25,6);
\node[text width=6cm, anchor=west, right] at (-1,5.6)   {$L^\eta$};

\foreach \x in {0,1,2,3,4,5}
\foreach \y in {0,1,2,3,4,5}
{
\draw[thick] (\x,\y) rectangle (\x+1,\y+1);
}
\foreach \x in {0,0.1,0.2,0.3,0.4,0.5,0.6,0.7,0.8,0.9,1,1.1,1.2,1.3,1.4,1.5,1.6,1.7,1.8,1.9,2,2.1,2.2,2.3,2.4,2.5,2.6,2.7,2.8,2.9,3,3.1,3.2,3.3,3.4,3.5,3.6,3.7,3.8,3.9,4,4.1,4.2,4.3,4.4,4.5,4.6,4.7,4.8,4.9,5,5.1,5.2,5.3,5.4,5.5,5.6,5.7,5.8,5.9,6}
\foreach \y in {0,0.1,0.2,0.3,0.4,0.5,0.6,0.7,0.8,0.9,1,1.1,1.2,1.3,1.4,1.5,1.6,1.7,1.8,1.9,2,2.1,2.2,2.3,2.4,2.5,2.6,2.7,2.8,2.9,3,3.1,3.2,3.3,3.4,3.5,3.6,3.7,3.8,3.9,4,4.1,4.2,4.3,4.4,4.5,4.6,4.7,4.8,4.9,5,5.1,5.2,5.3,5.4,5.5,5.6,5.7,5.8,5.9,6}
{\fill (\x,\y) circle (0.4pt);
}

\end{tikzpicture}
\caption{The case $\eta>\frac12$}

\end{figure}
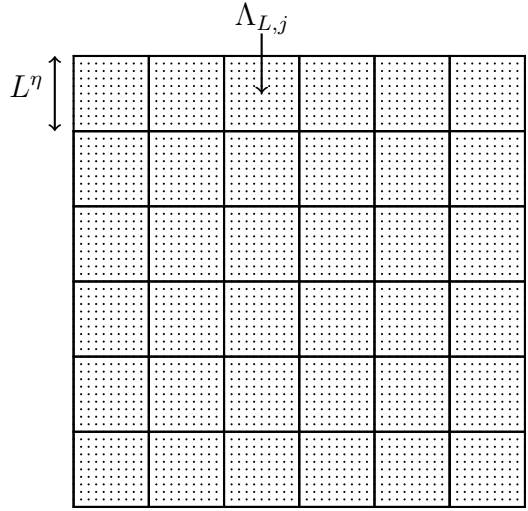

Define
$$Y_{L,j} =X_{L,j} - \Ex{X_{L,j}}.$$
We shall establish the convergence of the sum (over $j$) of $Y_{L,j}$ normalized by $\sqrt{|\Lambda_L|^{1-\tilde{\eta}}}$ to a Normal distribution using Theorem \ref{thm:Lindeberg}. Note that $Y_{L,j}$ are i.i.d., and each $X_{L,j}$ converges in distribution to a Poisson random variable with parameter $\lambda$. Applying Lemma \ref{lemmma:moments}, and Remark \ref{rem:uniform_convergence}, we obtain that
$$\Var{Y_{L,j}} \underset{L \rightarrow \infty}{\longrightarrow} \lambda^2,  $$
and this convergence is uniform in $j$.
Again, from the convergence of $X_{L,j}$ to a Poisson random variable and the uniform convergence of the variance to the variance of a Poisson random variable,
$$\underset{n \rightarrow \infty}{\lim} \frac{1}{\sum_j \Var{Y_{L,j}}}\sum_{j=1}^{M_L} \Ex{Y_{L,j}^2 \cdot \chi_{\left|Y_{L,j}\right| > \epsilon \sqrt{\Var{\sum_j Y_{L,j}} }}}=0 $$
for every $\epsilon >0$. Therefore,
$$\frac{\sum_{j=1}^{M_L(\beta)} X_{L,j}-\Ex{X_{L,j}}}{\sqrt{|\Lambda_L|^{1-\tilde{\eta}}}}\overset{d}{\longrightarrow} N(0,\sigma^2) $$
with $\sigma^2=f(E)^2\cdot (b-a)^2$. Note that just as in the previous proof, in order to determine the magnitude of the variance, we used the fact that
$$\frac{|\Lambda_L|^{1-\tilde{\eta}}}{M_L(\beta)}\overset{L \rightarrow \infty}{\longrightarrow} 1. $$
From Proposition \ref{prop:mainprop}, (taking $\beta = \tilde{\eta}$, $\alpha = \frac12-\frac{\tilde{\eta}}{2}$ and $g = \chi_{(a,b)}$)
\begin{equation}\label{equ:approx1}
  \underset{L\longrightarrow \infty}{\lim} \frac{1}{\sqrt{|\Lambda_L|^{1-\tilde{\eta}}}}\Ex{X_L-\sum_{j=1}^{M_L(\beta)} X_{L,j}}=0,
\end{equation}
and hence
$$\frac{X_L-|\Lambda_L|^{1-\tilde{\eta}}\cdot f(E)\cdot (b-a)}{\sqrt{|\Lambda_L|^{1-\tilde{\eta}}}}\overset{d}{\longrightarrow} N(0,\sigma^2). $$
From \eqref{equ:approx1}, we also infer that
$$\Var{\frac{X_L-\sum_{j=1}^{M_L\left(\tilde{\beta}\right)}X_{L,j}}{\sqrt{|\Lambda_L|^{1-\tilde{\eta}}}}}\overset{L \rightarrow \infty}{\longrightarrow} 0, $$
and since
$$\Var{\frac{\sum_{j=1}^{M_L\left(\tilde{\beta}\right)}X_{L,j}}{\sqrt{|\Lambda_L|^{1-\tilde{\eta}}}}} \overset{L \rightarrow \infty}{\longrightarrow} \sigma^2, $$
$$\Var{\frac{X_L}{\sqrt{|\Lambda_L|^{1-\tilde{\eta}}}}}=\Var{\frac{X_L-\sum_{j=1}^{M_L\left(\tilde{\beta}\right)}X_{L,j}}{\sqrt{|\Lambda_L|^{1-\tilde{\eta}}}}}
+\Var{\frac{\sum_{j=1}^{M_L\left(\tilde{\beta}\right)}X_{L,j}}{\sqrt{|\Lambda_L|^{1-\tilde{\eta}}}}}$$
$$+2\Cov{\frac{X_L-\sum_{j=1}^{M_L\left(\tilde{\beta}\right)}X_{L,j}}{\sqrt{|\Lambda_L|^{1-\tilde{\eta}}}},\frac{\sum_{j=1}^{M_L\left(\tilde{\beta}\right)}X_{L,j}}{{\sqrt{|\Lambda_L|^{1-\tilde{\eta}}}}}}
\overset{L \rightarrow \infty}{\longrightarrow} \sigma^2 $$
as well.
\end{proof}

Equipped with the CLT for large $\eta$'s, we can prove Theorem \ref{thm:main}.

\begin{proof} [Proof of Theorem \ref{thm:main}]
Let $\eta > 0$, and denote
$$\sigma^2 = f(E)\cdot (b-a)^2. $$
If $\eta >\frac12$, the theorem is immediately true using Proposition \ref{prop:localclt}. For $\eta \le \frac12$, there exists $j\in \N$ such that $\frac{1}{2^j} < \eta \le \frac{1}{2^{j-1}}$ and we begin by dividing $\Lambda_L$ into $M_L\left(\frac12\right)$ boxes $\{\Lambda_{L,k}\}_{k=1}^{M_L\left(\frac12\right)}$. Now for every $k$, we divide $\Lambda_{L,k}$ again with $\beta=\frac12$ into $M_{L,k}(\frac12)$ boxes $\{\Lambda_{L,k,k'}\}_{k'=1}^{M_{L,j}\left(\frac12\right)}$.
We iterate this process $j-1$ times and produce
$$M_L\left(\frac12\right)\cdot M_{L,k}\left(\frac12\right)\cdot... \cdot M_{L,\left(k_1,...,k_{j-1}\right)}\left(\frac12\right) $$
boxes (see Figure \ref{fig:smalleta}).

\begin{figure} \label{fig:smalleta}
\centering
\begin{tikzpicture}
\draw[thick] (0,0) rectangle (8,8);
\draw[ultra thick] (6,6) rectangle (8,8);

\draw[thick,<->] (-0.25,6) -- (-0.25,8);
\node[text width=6cm, anchor=west, right] at (-1,7)   {$L^\frac12$};
\draw[thick,->] (8.4,7) -- (7,7);
\node[text width=3cm, anchor=west, right] at (8.5,7)   {$\Lambda_{L,k}$};

\draw[ultra thick] (7,4) rectangle (8,5);
\draw[thick,<->] (-0.25,4) -- (-0.25,5);
\node[text width=6cm, anchor=west, right] at (-1,4.5)   {$L^\frac14$};
\node[text width=3cm, anchor=west, right] at (8.5,4.5)   {$\Lambda_{L,k_1,k_2}$};
\draw[thick,->] (8.4,4.5) -- (7.5,4.5);

\fill (-0.6,3.5) circle (1.2pt);
\fill (-0.6,3) circle (1.2pt);
\fill (-0.6,2.5) circle (1.2pt);

\fill (9.25,3.5) circle (1.2pt);
\fill (9.25,3) circle (1.2pt);
\fill (9.25,2.5) circle (1.2pt);

\foreach \x in {0,0.1,0.2,0.3,0.4,0.5,0.6,0.7,0.8,0.9,1,1.1,1.2,1.3,1.4,1.5,1.6,1.7,1.8,1.9,2,2.1,2.2,2.3,2.4,2.5,2.6,2.7,2.8,2.9,3,3.1,3.2,3.3,3.4,3.5,3.6,3.7,3.8,3.9,4,4.1,4.2,4.3,4.4,4.5,4.6,4.7,4.8,4.9,5,5.1,5.2,5.3,5.4,5.5,5.6,5.7,5.8,5.9,6,6.1,6.2,6.3,6.4,6.5,6.6,6.7,6.8,6.9,7,7.1,7.2,7.3,7.4,7.5,7.6,7.7,7.8,7.9,8}
\foreach \y in {0,0.1,0.2,0.3,0.4,0.5,0.6,0.7,0.8,0.9,1,1.1,1.2,1.3,1.4,1.5,1.6,1.7,1.8,1.9,2,2.1,2.2,2.3,2.4,2.5,2.6,2.7,2.8,2.9,3,3.1,3.2,3.3,3.4,3.5,3.6,3.7,3.8,3.9,4,4.1,4.2,4.3,4.4,4.5,4.6,4.7,4.8,4.9,5,5.1,5.2,5.3,5.4,5.5,5.6,5.7,5.8,5.9,6,6.1,6.2,6.3,6.4,6.5,6.6,6.7,6.8,6.9,7,7.1,7.2,7.3,7.4,7.5,7.6,7.7,7.8,7.9,8}
{\fill (\x,\y) circle (0.4pt);
}

\draw[thick,<->] (-0.25,0) -- (-0.25,0.5);
\node[text width=6cm, anchor=west, right] at (-1,0.25)   {$L^\eta$};
\draw[ultra thick] (7.5,0) rectangle (8,0.5);
\node[text width=3cm, anchor=west, right] at (8.5,0.25)   {$\Lambda_{L,k_1,...,k_{j}}$};
\draw[thick,->] (8.4,0.25) -- (7.75,0.25);

\foreach \x in {0,0.5,1,1.5,2,2.5,3,3.5,4,4.5,5,5.5,6,6.5,7,7.5}
\foreach \y in {0,0.5,1,1.5,2,2.5,3,3.5,4,4.5,5,5.5,6,6.5,7,7.5}
{
\draw[thick] (\x,\y) rectangle (\x+0.5,\y+0.5);
}

\end{tikzpicture}
\caption{The case $\eta \le \frac12$}
\end{figure}
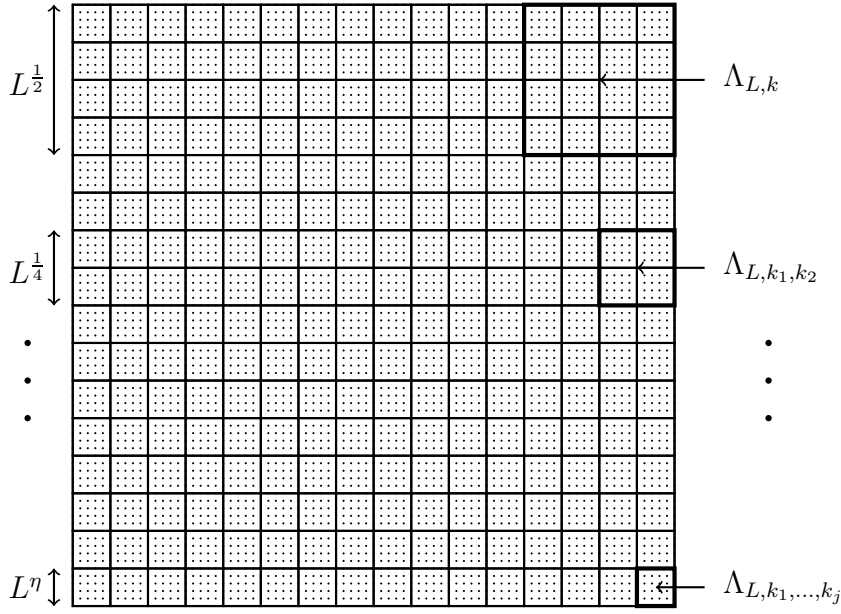

Note that
$$\frac{M_{L,\left(k_1,...,k_m\right)}(\frac12)}{|\Lambda_L|^{\frac{1}{2^{m+1}}}}\underset{L \rightarrow \infty}{\longrightarrow} 1, $$
so
$$\frac{M_L\left(\frac12\right)\cdot M_{L,k}\left(\frac12\right)\cdot... \cdot M_{L,\left(k_1,...,k_{j-1}\right)}\left(\frac12\right)}{|\Lambda_L|^{1-\frac{1}{2^{j-1}}}}\underset{L \rightarrow \infty}{\longrightarrow} 1 $$
and for any such box $\Lambda_{L,\left(k_1,...,k_{j-1}\right)}$, there exist
$$a_L \equiv a_L\left(k_1,...,k_{j-1}\right) \in \R^d, \ , c_L\equiv c_L\left(k_1,...,k_{j-1}\right)>0$$ and with $c_L \rightarrow 0$
such that
$$\Lambda_{L,(k_1,...,k_{j-1})} = \left([-L'+c_L,L'-c_L]^d+a_L\right) \cap \Z^d $$
with $L' = L^{\frac{1}{2^{j-1}}}$.
Thus, each $\Lambda_{L,(k_1,...,k_{j-1})}$ possesses the requirements of Proposition \ref{prop:localclt} with $\tilde{\eta}=2^{j-1}\eta$. This means that if we define the random variable $X_{L,(k_1,...,k_{j-1})}$ to be the number of eigenvalues of $H_{\Lambda_{L,(k_1,...,k_{j-1})}}$ in the interval $I_L$,
$$Y_{L,(k_1,...,k_{j-2},k_{j-1})} \equiv \frac{X_{L,(k_1,...,k_{j-1})}-\Ex{X_{L,(k_1,...,k_{j-1})}}}{\sqrt{|\Lambda_L|^{\frac{1}{2^{j-1}}-\eta}}}$$
converges to a Normal random variables with variance $\sigma^2$ as $L$ tends to infinity, and
\begin{equation}\label{equ:converge}
  \Var{Y_{L,(k_1,...,k_{j-2},k_{j-1})}} \overset{L \rightarrow \infty}{\longrightarrow} \sigma^2
\end{equation}
as well.
Moreover, using a similar argument to the one presented in Remark \ref{rem:uniform_convergence}, the convergence rate in \eqref{equ:converge} does not depend on $k_1,...,k_{j-1}$.

Fix $k_1,...,k_{j-2}$.
As in the proof of Proposition \ref{prop:localclt}, we apply Theorem \ref{thm:Lindeberg}.

$Y_{L,(k_1,...,k_{j-2},k_{j-1})}$ are independent,  $Y_{L,(k_1,...,k_{j-2},k_{j-1})}$ converges to a Normal random variable and
$$\Var{Y_{L,(k_1,...,k_{j-2},k_{j-1})}} \overset{L \rightarrow \infty}{\longrightarrow} \sigma^2. $$
Moreover, using a similar argument as in Remark \ref{rem:uniform_convergence}, this convergence does not depend on $k_1,...,k_{j-2},k_{j-1}$. Therefore, denoting $\tilde{k}_m=(k_1,...,k_{j-2},k_m)$
$$\frac{1}{|\Lambda_L|^{\frac{1}{2^{j-1}}}}\underset{L \rightarrow \infty}{\lim} \sum_{m=1}^{M_{L,\tilde{k}_m}} \Ex{Y_{L,\tilde{k}_m}^2 \cdot \chi_{\left|Y_{L,\tilde{k}_m}\right| > \epsilon \sqrt{\Var{\sum_j Y_{L,\tilde{k}_m}} }}}=0.$$
and thus
$$\sum_{k=1}^{M_{L,\left(k_1,...,k_{j-2}\right)}(\frac12)} Y_{L,(k_1,...,k_{j-2},k)} \underset{L\longrightarrow \infty}{\overset{d}{\longrightarrow}} N(0,\sigma^2)$$
From Proposition \ref{prop:mainprop},
$$\frac{X_{L,(k_1,...,k_{j-2})}-\Ex{X_{L,(k_1,...,k_{j-2})}}}{\sqrt{|\Lambda_L|^{\frac{1}{2^{j-2}}-\eta}}}\overset{d}{\longrightarrow} N(0,\sigma^2), $$
and
$$\Var{\frac{X_{L,(k_1,...,k_{j-2})}}{\sqrt{|\Lambda_L|^{\frac{1}{2^{j-2}}-\eta}}}}\overset{L \rightarrow \infty}{\longrightarrow} \sigma^2. $$
Again, the rate of convergence here does not depend on $k_1,...,k_{j-2}$.

Iterating this process $j-2$ more times, we end with the desired result.

\end{proof}

\end{document}